\documentclass[12pt]{article}
\usepackage{amssymb,amsmath,amscd,amsthm}

\newtheorem{theorem}{Theorem}[section]

\newtheorem{lemma}[theorem]{Lemma}

\date{}

\begin{document}

\title{Sharp Weyl Law  for Signed Counting Function of Positive Interior Transmission Eigenvalues}

\author{ E.Lakshtanov\thanks{Department of Mathematics, Aveiro University, Aveiro 3810, Portugal.  The work was supported by Portuguese funds through the CIDMA - Center for Research and Development in Mathematics and Applications and the Portuguese Foundation for Science and Technology (``FCT--Fund\c{c}\~{a}o para a Ci\^{e}ncia e a Tecnologia''), within project PEst-OE/MAT/UI4106/2014 (lakshtanov@ua.pt).} \and
 B.Vainberg\thanks{Department
of Mathematics and Statistics, University of North Carolina,
Charlotte, NC 28223, USA. The work was partially supported  by the NSF grant DMS-1008132 (brvainbe@uncc.edu).}}

\maketitle

\begin{abstract}
We consider the interior transmission eigenvalue (ITE) problem that arises when scattering by inhomogeneous media is studied. The ITE problem is not self-adjoint. We show that positive ITEs are observable together with plus or minus signs that are defined by the direction of motion of the corresponding eigenvalues of the scattering matrix (as they approach $z=1$). We obtain a Weyl-type formula for the counting function of positive ITEs, which are taken together with the ascribed signs. The results are applicable to the case when the medium contains an unpenetrable obstacle.
\end{abstract}



\section{Main results}

Let $\mathcal O\subset R^d$ be an open bounded domain with $C^2$ boundary $\partial O$ and the outward normal $\nu$.
\textit{The interior transmission eigenvalues (ITEs)} are defined as the values of $\lambda \in \mathbb C$
for which the problem
\begin{eqnarray}\label{Anone0}
-\Delta u - \lambda u =0, \quad x \in \mathcal O, \quad u\in H^2(\mathcal O), \\ \label{Anone}
-\Delta v - \lambda   n(x)v =0, \quad x \in \mathcal O, \quad v\in H^2(\mathcal O), \\ \label{Antwo}
\begin{array}{l}
u-v=0, \quad x \in \partial \mathcal O, \\
\frac{\partial u}{\partial \nu} - \frac{\partial v}{\partial \nu}=0, \quad x \in \partial \mathcal O,
\end{array}
\end{eqnarray}
has a non-trivial solution.
Here $n(x)>0,$ $ x\in \overline{\mathcal O}$, is a smooth positive function, $H^2(\mathcal O)$ is the Sobolev space. Only real positive ITEs will be considered below.

This spectral problem for the system of two equations in a bounded domain $\mathcal O\subset \mathbb R^d$ appears naturally when the scattering transmission problem (scattering of plane waves by an inhomogeneous medium) is studied. The scattering
problem (for $\lambda>0$) is stated as follows:
\begin{equation}\label{AnoneE}
 -\Delta \psi - \lambda  \widehat{n}(x)\psi =0,\quad x \in \mathbb R^d,
\end{equation}
where $\widehat{n}(x)=n(x),~~ x\in \overline{\mathcal O}; \quad \widehat{n}(x)=1, ~~ x\in \mathbb R^d\backslash\overline{\mathcal O},$ and $\psi$ is the sum of the incident plane wave and the scattered wave, i.e., $\psi=e^{ik\omega \cdot x}+\psi_{\text{sc}},~\lambda=k^2,~~\psi_{\text{sc}}$ satisfies the radiation conditions:
$$
\psi_{\text{sc}}=f(k,\theta,\omega)\frac{e^{ikr}}{r^{\frac{d-1}{2}}}  + O \left (r^{-\frac{d+1}{2}} \right ),\quad \theta=\frac{x}{r},~~r=|x|\to\infty.
$$

We also consider the case where $\mathcal O$ contains a compact obstacle $\mathcal V \subset \mathcal O, \partial \mathcal V \in C^2$ ($\mathcal V$ can be a union of a finite number of obstacles). Then equations (\ref{Anone}) and (\ref{AnoneE}) are replaced by similar equations in $\mathcal O \backslash \mathcal V,~~\mathbb R^d\backslash \mathcal V,$ respectively, with the Dirichlet or Neumann boundary condition on $\partial \mathcal V$.
For the sake of clarity, the proofs of all the statements will be given when $\mathcal V=\emptyset$. The case of $\mathcal V\neq\emptyset$ is addressed in Remark 3.

The main relation between the scattering and ITE problems is due to the following fact: if the far-field operator
\begin{equation}\label{farf}
F=F(k):L_2(S^{d-1}) \rightarrow L_2(S^{d-1}),~~F\phi=\int_{S^{d-1}}f(k,\theta,\omega)\phi(\omega)dS_\omega
\end{equation}
has zero eigenvalue at the frequency $k=k_0>0$, then $\lambda=k_0^2$ is an ITE. The proof of this fact is very simple and can be found in \cite{K86}.
 This relation between {\it positive} ITEs and operator $F$ is very important in the study of scattering by inhomogeneous media. In particular,
it has been extensively used in the linear sampling and the factorization methods of inverse scattering, starting with papers \cite{K86},\cite{cm},\cite{RSleeman}.

 One can consider the scattering matrix of the transmission problem (\ref{AnoneE}) instead of the far-field operator $F$. It is the unitary operator given by
\begin{equation}\label{sk}
S(k)= I + 2ik \overline{\alpha}  F ~ : ~ L_2(S^{d-1}) \rightarrow L_2(S^{d-1}), \quad \alpha=\frac{1}{4\pi}\left (\frac{k}{2\pi i} \right )^{\frac{d-3}{2}}.
\end{equation}
Thus, the existence of the eigenvalue $z=1$ of the operator $S(k)$ for some $k=k_0$ implies that $\lambda=k_0^2$ is an ITE.

While the latter statement (the existence of eigenvalue $z=1$ of  $S(k_0)$ implies that $\lambda=k_0^2$ is an ITE) is rather simple, the converse relation is much more delicate. Both relations together are called the weak inside-outside duality principle\footnote{This was studied in \cite{EP} for scattering by an obstacle and in \cite{kirschL} for the transmission problem.}. Roughly speaking, it says that an eigenvalue $z(k)$ of the scattering matrix $S(k)$ can have a one-sided limit $z=1$ as $k\rightarrow k_0$ if and only if $k_0^2$ is an ITE. As a byproduct of our main result, this principle will be justified below under minimal assumption on $n(x)$.

The eigenvalues $z_j(k)$ of the unitary operator $S(k)$ belong to the unit circle $C=\{z:|z|=1\}$. The operator $F$ is compact and its eigenvalues converge to zero, and therefore for each $k>0$ the eigenvalues $\{z_j(k)\}$ converge to $z=1$ as $j\to\infty$. Functions $z_j(k)$ are not always analytic at points $k$ where $z_j=1$ ($z=1$ is the essential point of the spectrum of $S(k)$). For example, $z=1$ is never an eigenvalue of $S(k)$ if $\mathcal O$ has a corner, see \cite{BPaiv2}. It is known in a quite general setting (see \cite{birman}, \cite{kato}) that the eigenvalues $z_j(k)$ are distributed in a neighborhood of the point $z=1$ very non-uniformly. One of the half-circles $C_\pm=C\bigcap \{\pm \Im z > 0\}$ usually contains at most a finite number of the eigenvalues, while the opposite half-circle contains infinitely many of them (with the limiting point at $z=1$).

Let us describe the inside-outside duality principle in the case of scattering by a soft or rigid obstacle $\mathcal O$ (see \cite{DorSmil},\cite{Smil},\cite{EP},\cite{EP2}). In this case, the interior Dirichlet or Neumann problem is considered instead of the interior transmission problem, and the eigenvalues of the Dirichlet or Neumann Laplacian in $\mathcal O$ are used instead of ITEs. The half-circle $C_+$ contains a finite number of the eigenvalues $\{z_j(k)\}$ in the case of the Dirichlet condition. If $\lambda=k_0^2$ is an eigenvalue of the interior Dirichlet problem of multiplicity $m$, then exactly $m$ points $z_j(k)$ from $C_+$ converge to $z=1$ as $k\to k_0-0$. These points move clockwise. A similar statement is valid in the case of the Neumann boundary condition: the half-circle $C_-$ contains a finite number of the eigenvalues $\{z_j(k)\}$ in this case, and $m$ of them converge to $z=1$ as $k\to k_0+0$. The latter eigenvalues also move clockwise with the increase of $k$.

This clockwise motion of all the eigenvalues $z_j(k)$ clarifies the earlier results on high frequency asymptotics of the total phase of $S(k)$ established in the case of scattering by an obstacle (see \cite{jk}, \cite{mr}):
\begin{equation}\label{argdet}
\arg \det S(k)=- \frac{\omega_d|\mathcal O|}{(2\pi)^{d-1}}\lambda^{\frac d 2}(1+\lambda^{-\varepsilon}),  \quad \lambda=k^2\to\infty,
\end{equation}
where $\omega_d$ is the volume of the unit ball in $R^d$ and $|\mathcal O|=\operatorname{Vol}(\mathcal O)$.
Indeed, due to the Weyl formula for the counting function of interior eigenvalues, the right-hand side above coincides with the number of eigenvalues on the interval $(0,\lambda)$ multiplied by $-2\pi$, i.e., it is equal to the sum of all the phases that make the total rotation (always moving clockwise) around the circle $|z|=1$.

The ITE problem was studied in \cite{kirschL}. The clockwise convergence of an eigenvalue $z(k)$ of the scattering matrix $S(k)$ to $z=1$ was proved there under certain conditions on $n$ when $k$ approaches the smallest ITE $k=k_0$ from the appropriate side of $k_0$. The main coefficient in formula (\ref{argdet}) for the transmission problem should have the following property: it should decrease by $2\pi$ every time when $z_j(k)$ reaches $z=1$  while moving clockwise, and it should increase by $2\pi$ when $z_j(k)$ reaches $z=1$  while moving counterclockwise. This is one of the reasons why we introduce and study the signed Weyl formula for ITEs. Other reasons are discussed after the statement of Theorem \ref{th}.


The main result of this paper is as follows.
We will ascribe a value $\sigma_i=\pm 1$ to each simple {\it positive} ITE. Moreover, these values are observable and correspond to the clockwise (counterclockwise, respectively) motion of the eigenvalue of the scattering matrix toward $z=1$. In the case of an ITE of geometric multiplicity $m>1$, we ascribe a coefficient $\sigma_i, |\sigma_i|\leq m$, to the whole group, not to each of the ITEs separately. We will specify $\sigma_i$ in more detail later. Then we count {\it positive} eigenvalues together with the sign ascribed.
\begin{theorem}\label{th}
Let $n(x) \neq 1, x \in \partial \mathcal O$ (i.e., the inhomogeneity in the scattering problem has a sharp boundary). Then the Weyl law holds for the signed counting function of the interior transmission eigenvalues:
\begin{equation}\label{theorem2}
\sum_{i ~ : ~ 0< \lambda^T_i < \lambda} \sigma_i = \frac{\omega_d}{(2\pi)^d}   \gamma   \lambda^{\frac{d}{2}} + O(\lambda^{\frac{d}{2}-\delta}), \quad \lambda \rightarrow \infty, \quad \delta=\frac{1}{2d}, \mbox{ where }
\end{equation}
\begin{equation}\label{eqnew}
\gamma:= \operatorname{Vol}(\mathcal O) - \int_{\mathcal O \backslash \mathcal V}n^{\frac d 2}(x)dx \neq 0.
\end{equation}
\end{theorem}
We expect that a similar result is valid for other scattering problems related to Maxwell and Dirac equations, scattering on graphs, etc.

Let us stress that the asymptotics of the positive ITE spectrum with the signs $\sigma_i=\pm 1$ ascribed determines $\operatorname{Vol}(\mathcal O) - \int_{\mathcal O \backslash \mathcal V}n^{\frac d 2}(x)dx$ similarly to the situation with the Weyl asymptotics that determines the volume of the domain in the case of the Dirichlet or Neumann problem. Among other applications of Theorem \ref{th}, let us mention the justification of an analogue of formula (\ref{argdet}) for the transmission problem with a potential that has a jump on $\partial\mathcal O$.  The constant $|\mathcal O|$ in (\ref{argdet}) will be replaced by $\gamma$ in this case. A similar result for operators with infinitely smooth coefficients (without a jump on $\partial\mathcal O$) can be found in \cite{Robert}. We plan to discuss all these applications in a future publication.

The main difficulty in the proof of Theorem \ref{th} is related to the fact that problem (\ref{Anone0})-(\ref{Antwo}) is neither elliptic nor symmetric. Indeed, the Shapiro-Lopatinski conditions must hold at the boundary for the ellipticity of the problem. To convince a reader that these conditions are violated, without making any calculations, one can note that the ellipticity of a problem implies that its null space is finite dimensional. Additionally, if (\ref{Anone0})-(\ref{Antwo}) is elliptic, then the same problem with $n(x)=1$ would be elliptic. But the latter problem has an infinitely-dimensional kernel that contains all $(u,v)$ such that $u=v$. The violation of the symmetry can also be very easily checked.

The lack of symmetry and ellipticity makes the study of ITEs much more difficult than the study of eigenvalues of the Dirichlet or Neumann Laplacian. In particular, the discreteness of the spectrum of the ITE problem, the existence of real eigenvalues, and their asymptotics can not be obtained by soft arguments. Moreover, the existence of non-real ITEs was shown in \cite{nonreal}, and an example of an elliptic ITE problem where the set of ITEs is not discrete can be found in \cite[Examples 1,2]{LV4}.

{\bf Overview of previous results on ITE}.
There is extensive literature (see the review \cite{HadCak}) on the properties of ITEs and corresponding eigenfunctions. The following results are most closely related to our study. It was shown in \cite{sylv} that the set of ITEs is discrete if $n(x) \neq 1$ everywhere at the boundary of the domain $\partial \mathcal O$. The latter condition (which means that the inhomogeneity has a sharp boundary) will be assumed to hold in our study. It was shown in \cite{BPaiv},\cite{Petkov},\cite{Faierman},\cite{lakvain5},\cite{Rob2}\footnote{Paper \cite{lakvain5} concerns the anisotropic ITE problem, papers \cite{BPaiv},\cite{Faierman} concern the case of $n(x)>1, x \in \overline{\mathcal O}$, and papers \cite{Petkov},\cite{Rob2} concern the case of $n(x) \neq 1, x \in \partial \mathcal O$.} that the standard  Weyl estimate holds for the {\it complex} ITEs located in
 an arbitrary cone containing the real positive semi-axis:
\begin{equation}\label{eqold}
\#\{i:|\lambda^T_i| \leq \lambda\} = \frac{\omega_d}{(2\pi)^d}  \left [\operatorname{Vol}(\mathcal O) + \int_{\mathcal O} n^{\frac d 2}(x) dx \right ] \lambda^{\frac{d}{2}}+ o(\lambda^{\frac d 2}), \quad \lambda \rightarrow \infty,
\end{equation}
where $\omega_d$ is the volume of the unit ball in $\mathbb R^d$.

Earlier, in a series of articles \cite{lakvain6},\cite{lakvain7},\cite{lakvain8}, we have shown that if $\gamma\neq 0$,
then the set of  {\it positive} ITEs (which are the ones important in applications) is infinite, and, moreover,
\begin{equation}\label{theorem}
\#\{i:0< \lambda^T_i < \lambda\} ~ \geq ~ \frac{\omega_d}{(2\pi)^d} |\gamma|   \lambda^{\frac{d}{2}} + O(\lambda^{\frac{d}{2}-\delta}), \quad \lambda \rightarrow \infty.
\end{equation}
Obviously, the coefficient in the main term above is always smaller than the corresponding coefficient in (\ref{eqold}).

The first result of the present paper shows that one of the half-circles $C_\pm$ contains at most a finite number of the eigenvalues of $S(k)$ in the case of the transmission problem. It is an extension of the result from \cite{kirschL}, where it was assumed that $n(x)\neq 1,~x\in\overline{\mathcal O}$. Our result is proved under the weaker assumption that $n(x)\neq 1$ only at the boundary. We also will show that if $n-1$ changes sign on the boundary, then both half-circles contain infinitely many eigenvalues of the scattering matrix. Namely, the following theorem will be proved.
\begin{theorem}\label{lemmaup}
\begin{enumerate}
\item If $n(x) < 1, x \in \partial \mathcal O$, then $S(k)$ has at most a finite number of eigenvalues $z_j(k)$ in  $C_+$ for each fixed $k>0$ (as in the case of the Dirichlet boundary condition).
\item If $n(x) > 1, x \in \partial \mathcal O$, then $S(k)$ has at most a finite number of eigenvalues $z_j(k)$ in  $C_-$ for each fixed $k>0$ (as in the case of the Neumann boundary condition).
\item If $n(x)-1$ takes both positive and negative values on $\partial \mathcal O$, then  $S(k),k>0,$ has infinitely many eigenvalues in both  $C_+$ and  $C_-$.
\end{enumerate}
\end{theorem}

After we prove the theorem, we will study the ITE problem under the first two assumptions of the theorem. We will concentrate our attention on the half-circle with a finite number of points $z_j(k)$. Our next goal is to study the motion of the points $z_j(k)$ along the half-circle when $k$ is increasing. We distinguish clock- and counterclockwise motion of these points.  We are primarily interested in what happens when a point $z_j(k)$ reaches $z=1$ at some moment $k=k_0$, for example when $k\to k_0-0$.  When $k>k_0$, the point may stay on the same half-circle, move to another half-circle, or disappear. It is difficult to find a proper description of all the possibilities and justify them due to the lack of smoothness of the eigenvalues at the points where  $z_j(k)=1$. Thus we will split the motion of each point $z_j(k)$ into two parts: before $z_j(k)$ reaches $z=1$ at $k=k_0$ and after that event, without any attempt to relate the eigenvalues before $k_0$ and after $k_0$.

Denote by $m^+=m^+(k_0)~(m^-=m^-(k_0))$ the number of eigenvalues $z_j(k)$ of the scattering matrix $S(k)$ that are at the point $z=1$ at the moment $k=k_0$ while $z_j(k)$ are moving in the chosen half-circle clockwise (counterclockwise, respectively). For example, if $n(x)<1$ on $\partial\mathcal O$, then the chosen half-circle is $C_+$, and $m^\pm$ is the number of eigenvalues such that
\[
\lim_{k\to k_0\mp 0}z_j(k)=1+i0.
\]

We will prove
\begin{theorem}\label{t3} 1) Let  $n(x) < 1,~ x \in \partial \mathcal O$.
 For each ITE
$\lambda_i^T=k_i^2$ of geometric multiplicity $m_i$, there are $m^+_i\leq m_i$ eigenvalues $z_j(k)$ of the scattering matrix in the upper half-circle $C_+$ that approach $z=1$  when $k\to k_i-0$ (they are moving clockwise when $k$ increases) and $m^-_i\leq m_i$ eigenvalues that approach $z=1$  when $k\to k_i+0$ (they are moving counterclockwise when $k$ increases). There is an arc of the unit circle defined by $0<\arg z<\delta$ that is free of all other points $z_j(k), ~ |k-k_i| \ll 1.$

2) Let  $n(x) > 1,~ x \in \partial \mathcal O$.
 For each ITE
$\lambda_i^T=k_i^2$ of geometric  multiplicity $m_i$, there are $m^+_i\leq m_i$ eigenvalues $z_j(k)$ of the scattering matrix in the lower half-circle $C_-$ that approach $z=1$  when $k\to k_i+0$ and $m^-_i\leq m_i$ eigenvalues that approach $z=1$  when $k\to k_i-0$. There is an arc of the unit circle defined by $-\delta<\arg z<0$ that is free of all other points $z_j(k), ~ |k-k_i| \ll 1.$

3) The statements above remain valid if $\lambda=k_0^2>0$ is not an ITE (i.e., $m_i=0$). In this case, the corresponding arc of the unit circle is free of the eigenvalues of the scattering matrix when $|k-k_0| \ll 1.$
\end{theorem}
Note that $m_i^\pm$ can take more or less arbitrary integer values in the segment $[0,m_i]$. In particular, the relation $m_i^++m_i^-=m_i$ does not necessarily hold. However, one can derive from the arguments in the proof that $m_i^++m_i^-=m_i ({\rm mod} 2)$. In particular, if $m_i=1$, i.e., an ITE is simple, then one of the numbers $m_i^\pm$ is one and another is zero. A more specific result is given by the following theorem.
\begin{theorem}\label{th14} If $\lambda_i^T=k_i^2$ is an ITE whose geometric and algebraic multiplicities coincide (there are no adjoint eigenfunctions) and $(U_j,V_j),~1\leq j\leq i,$ is a basis in the eigenspace, then $m^+_i-m^-_i$ is the signature of the following matrix $A=\{a_{j,l}\}$:
$$
m^+_i-m^-_i =  {\rm sgn} \left  \{a_{j,l} \right \}, \quad a_{j,l}=\int_{\mathcal O} \left (U_jU_l- nV_jV_l  \right ) dx,~~1\leq j,l\leq m_i.
$$
\end{theorem}
We already mentioned above that Theorem \ref{t3} was proved in \cite{EP},\cite{EP2} in the case of scattering by a soft or rigid obstacle $\mathcal O$. In the latter case, the interior Dirichlet or Neumann problem is considered instead of the interior transmission problem, and the eigenvalues of the Dirichlet or Neumann Laplacian in $\mathcal O$ are used instead of ITEs. The symmetry and simplicity of the corresponding interior problem imply that all the scattering eigenvalues in both cases (Dirichlet/Neumann condition) move in the same clockwise direction. Moreover, in both cases, $m^+_i=m_i,~ m^-_i=0$ for all $i\geq 1$.  Perhaps this explains why the signed Weyl formulas did not appear in the literature earlier.

Some results related to Theorem \ref{th14} can be found in  \cite{kirschL}.  If $n(x)\neq 1,~x\in \overline{\mathcal O}$, and certain conditions on an ITE $\lambda_i^T=k_i^2$ hold, it was proved there that the eigenvalue $z_j(k)$ closest to $z=1$ and located in the appropriate half-circle converges to $z=1$ when $k$ approaches $k_i$ from the appropriate side.

{\it Our last and main result is as follows: Theorem \ref{th} holds with $\sigma_i$ defined by the formula}
\begin{equation}\label{defsi}
\sigma_i = m^+_i - m^-_i.
\end{equation}


{\bf Remark 1}. There are a couple of cases when formula (\ref{theorem2}) can be obtained by direct calculations. The asymptotics of the counting function for positive ITE  was found in \cite{stefanov} when $\mathcal O$ is a ball and $n(x)$ is a constant. It was shown in \cite{stefanov} that the counting function of positive ITEs coincides with the right-hand side in (\ref{theorem2}). Since the Weyl law and signed Weyl law have the same main term, this implies that the majority of $\sigma_i$ have the same sign, which is equal to ${\rm sign}(1-n)$. It also follows from calculations in that paper (section 3.3) that $\sigma_i{\rm sign}(1-n)=1$ for the problem under consideration if the ITE is simple.

Another example concerns  problem (\ref{Anone0})-(\ref{Antwo}), where $n$ is replaced by $n/a$ and the second boundary condition is replaced by $\frac{\partial u}{\partial \nu} = a \frac{\partial v }{\partial \nu}$ (models of this type were introduced in \cite{cch}). Note that the results of the present paper hold for this problem, and the proofs require only minor changes. In a very trivial situation of $a=n=\rm{const} \neq 1$, the spectrum of the ITE problem is a union of Dirichlet and Neumann spectra for the negative Laplacian. So, all the ITEs are real, and (\ref{eqold}) provides the asymptotics for their counting function. Obviously, $\gamma=0$ in this case. So we see that the asymptotics for the counting function of positive ITEs and the signed counting function for positive ITEs are different in this simple case.

{\bf Remark 2}. The signs ascribed to ITEs in Theorem \ref{th} resemble the standard procedure used in the definition of spectral flows (see, for example, \cite{apz}, \cite{push}, \cite{saf}, \cite{saf1}). The principal difference is that these signs in the present paper are defined not by the direction of motion of the eigenvalues of the operator under consideration, but by the direction of motion of another object, namely, the eigenvalues of the scattering matrix.

While there is no direct reformulation of the signed sum of ITEs through spectral flows, the latter are relevant to the problem under investigation. Indeed, we reduced the original rather complicated non-self-adjoint (and non-elliptic) problem to a problem on the spectral flow for a symmetric Fredholm operator $R(\lambda)$ that has a simple representation (see (\ref{r2r})) via the Dirichlet-to-Neumann maps $N^{in}_0(\lambda) ,~N^{in}_n(\lambda)$, defined by equations (\ref{Anone0}), (\ref{Anone}). Operator $R(\lambda)$ depends meromorphically on the spectral parameter $\lambda$, and there are no readily available methods to calculate its spectral flow.

To be more precise, the signed Weyl law for ITEs is defined by the asymptotics (as $\lambda'\to\infty $) of the number $n_2(\lambda')$ of positive values of $\lambda<\lambda'$ for which $R(\lambda)$ has a non-trivial kernel. Each such value $\lambda=\lambda_0$ is counted with the sign that depends on the direction in which the corresponding eigenvalue $\mu_j(\lambda)$ of operator  $R(\lambda)$ passes through the origin $\mu=0$ when $\lambda $ increases and passes through $\lambda_0$. Since we can not evaluate $n_2(\lambda')$ directly, we use the conservation law ${n^-}(\lambda')=n_1(\lambda')+n_2(\lambda')$, where $n^-(\lambda)$ is the total number of negative eigenvalues $\mu_j(\lambda)$ and $n_1(\lambda')$ is the signed counting function for eigenvalues $\mu_j(\lambda)$ that enter/exit the negative semi-axis $R_\mu^-$ through the point $\mu=-\infty$. We evaluate $n^-(\lambda)$ and find the asymptotics of $n_1(\lambda')$. Note that changes in $n_1(\lambda')$ occur when $\lambda$ passes through poles of $R(\lambda)$. Perhaps, this type of contribution is not very standard in spectral flows (paper by Friedlelnder \cite{fried} was a trigger point for us in evaluating $n_1(\lambda')$). However, the most important part of the present paper concerns not the asymptotics of $n_2$ but  the relation between $n_2(\lambda')$ and the directions of rotation of the eigenvalues of the scattering matrix.

{\bf Remark 3}. The presence of an obstacle $\mathcal V \subset\mathcal O$ affects only the proof of Theorem \ref{th2a}. To be more precise, relations (\ref{44}), (\ref{n33}) must be proved in the presence of the obstacle ($\gamma$ in (\ref{44}) will depend on $\mathcal V$).  The corresponding formulas can be found in \cite{lakvain7}.

{\bf Conjecture.} We believe that Theorem \ref{th14} remains valid without the assumption on the absence of the adjoint eigenvectors. One only needs to construct matrix $A$ using a basis $(u_i,v_i)$ in the root subspace that corresponds to the ITE $\lambda_i^T$. Supporting arguments will be provided after the proof of Theorem \ref{th14}.

The rest of the paper is organized as follows. Section 2 concerns the relationships between Dirichlet-to-Neumann maps, ITEs, and the far-field operator. It starts (part (A)) with certain properties of the Dirichlet-to-Neumann maps and their relation to ITEs, followed (part (B)) by a representation of the far-field operator $F$ via a combination of the Dirichlet-to-Neumann operators. While one of the factors in this representation of $F$ is related to ITEs, this connection between $F$ and ITEs will be deepened. Additional properties of the Dirichlet-to-Neumann maps are obtained in part (C), followed by an alternative way to define ITEs (part (D)) and an alternative representation of the far-field operator $F$ (part (E)). In particular, a Weyl type formula for a signed sum of ITEs similar to (but different from) that in Theorem \ref{th} is proved in Theorem \ref{th2a} of part (D).

Section 3 completes the proofs of the main theorems. Theorem \ref{lemmaup} is proved first. The proof of Theorem \ref{t3} starts with general facts on quadratic forms defined by unitary operators. Then it is shown that  Theorem \ref{t3} holds with $m^\pm_i$ such that $m^+_i-m^-_i=\alpha^+_i-\alpha^-_i$. The latter relation together with Theorem \ref{th2a} imply Theorem \ref{th}. Theorem~\ref{th14} is proved at the end of the section.

\section{Relations between the scattering matrix and ITEs; auxiliary lemmas.}\label{sectionABCD}
{\it (A). Dirichlet-to-Neumann operators and their relationship with ITEs.} The operators
\begin{equation}\label{NN}
N^{in}_0,~N^{in}_n,~ N^{out}:H^s(\partial \mathcal O)\to H^{s-1}(\partial \mathcal O)~~ {\rm and}~~ N^{in}_0 -N^{in}_n:H^s(\partial \mathcal O)\to H^{s+1}(\partial \mathcal O)
\end{equation}
will be used heavily to prove the main result. Here $H^s(\partial \mathcal O)$ is the Sobolev space, $N^{in}_n(\lambda)$ is the Dirichlet-to-Neumann map for equation (\ref{Anone}) in $\mathcal O$, which is defined as follows:
\begin{equation}\label{DN}
N^{in}_n(\lambda): v\Big|_{\partial\mathcal O}\to\Big.\frac{\partial v}{\partial\nu}\Big|_{\partial\mathcal O},
\end{equation}
$N^{in}_0(\lambda)$ is the same operator for equation  \eqref{Anone0}, and $N^{out}(\lambda)$ is the Dirichlet-to-Neumann map for equation \eqref{Anone0} outside $\mathcal O $ that maps the Dirichlet data to the Neumann data of the solution that satisfies the radiation condition at infinity:
\[
u_r-iku=o(r^{(1-d)/2}), \quad  r\to\infty, \quad k^2=\lambda.
\]
We use the same normal vector $\nu$ for the exterior problem as that chosen (see equation (\ref{Antwo})) for the interior problem.

Let $(y_1,...y_{d-1})$ be local coordinates on $\partial \mathcal O$ with the dual variables $(\xi_1, ...,\xi_{d-1})$ and let $\sum g_{i,j}(y)dy_idy_j$ be the first fundamental form on $\partial \mathcal O$. Then $|\xi^*|=(\sum g^{i,j}(y)\xi_i\xi_j)^{1/2}$ is the length of the covector in the cotangent bundle $T^*(\partial\mathcal O)$.
\begin{lemma}\label{mprl}
1) The first two operators in (\ref{NN}) are self-adjoint elliptic pseudo-differential operators of first order. They are meromorphic in $\lambda $ when $\lambda>0$ with poles at eigenvalues of the Dirichlet problem for equations (\ref{Anone0}) and (\ref{Anone}), respectively. Their residues have finite ranges. The principal symbols of these two operators are equal to $|\xi^*|$.

2) Operator $N^{out}$ is an analytic in  $\lambda ,\lambda>0,$ elliptic pseudo-differential operator of first order with the principal symbol $-|\xi^*|$. For every $\varphi\neq 0$,
\begin{equation}\label{kvf}
\Im(N^{out}\varphi,\varphi)=\sqrt{\lambda}\int_{|\theta|=1}|f(\lambda,\theta)|^2dS>0,
\end{equation}
where $f$ is the far-field amplitude of the solution of the exterior problem with the Dirichlet data $\varphi$.

3) The last operator in (\ref{NN}) is a meromorphic in $\lambda $ elliptic pseudo-differential operator of order $-1$ with the principal symbol $\frac{\lambda(n(x)-1)}{2|\xi^*|}$.
\end{lemma}

 \begin{proof} The first statement and the expression for the symbol of $N^{out}$ are well known, see more details in \cite{lakvain7}. Formula (\ref{kvf}) is a direct consequence of the Green formula. The positivity of (\ref{kvf}) and analyticity of  $N^{out}$ are due to the absence of eigenvalues of the exterior Dirichlet problem imbedded into the continuous spectrum. The last statement can be found in \cite[Lemma 1.1]{lakvain7}. It is justified by calculating the first three terms of the full symbol of operators $N^{in}_0$ and $N^{in}_n$ (the first two terms of the symbols are canceled when the difference is taken).
\end{proof}

Note that the notion \emph{kernel of an operator} is used below not only when the operator is analytic in $\lambda$, but also when it has a pole at $\lambda=\lambda_0$. In the latter case, the kernel is understood as follows.
\\

\noindent{\bf Definition.}
{\it The kernel of a meromorphic operator function is the set of elements that are mapped to zero by both the analytical and the principal part of the operator.}
\\

The following lemma is a direct consequence of the definition of ITEs.
\begin{lemma}\label{ites}\textnormal{\cite{lakvain6},\cite{lakvain7}}
A point $\lambda=\lambda_0$ is an ITE if and only if the operator $N^{in}_0(\lambda) -N^{in}_n(\lambda)$ has a non-trivial kernel at  $\lambda=\lambda_0$ or
the following two conditions hold:

1) $\lambda=\lambda_0$ is an eigenvalue of the Dirichlet problem for $-\Delta$ and for equation (\ref{Anone}), i.e.,   $\lambda=\lambda_0$ is a pole for
both $N^{in}_0(\lambda)$ and $N^{in}_n(\lambda)$.

2) The ranges of the residues of operators $N^{in}_0(\lambda)$ and $N^{in}_n(\lambda)$ at the pole $\lambda=\lambda_0$ have a non-trivial intersection.

Moreover, the multiplicity of the interior transmission eigenvalue $\lambda=\lambda_0$ in all the cases is equal to $m_1+m_2$, where $m_1$ is the dimension of the
kernel of the operator $N^{in}_0(\lambda)-N^{in}_n(\lambda)$, and   $m_2$ is the   dimension of the intersection of the ranges of the residues  at the pole
$\lambda=\lambda_0$ ($m_2=0$ if $\lambda=\lambda_0$ is not a pole).
\end{lemma}

If $\lambda=\lambda_0$ satisfies the latter two conditions, it will be called a singular ITE. Thus, singular ITEs belong to
the intersection of three spectral sets: $ \{\lambda^{T}_i\}$ and sets of the eigenvalues of the Dirichlet problems for $-\Delta$ and for equation (\ref{Anone}).

{\it (B). A relation between D-to-N operators (\ref{NN}) and the far-field operator.} We will describe below some properties of operators (\ref{NN}), but first let us formulate some relationships of these operators with the scattering matrix and with ITEs that will allow us to relate the latter objects. The relation to $S(k)$ is given by the following statement.
\begin{theorem}\label{t21}
Let operator $\mathcal L : L_2(S^{d-1}) \rightarrow L_2(\partial \mathcal O)$ be defined by
\begin{equation}\label{lstar}
(\mathcal L \varphi)(x)=\int_{S^{d-1}} e^{ikx\cdot \omega} \varphi(\omega) dS_\omega, \quad (\mathcal L^* u)(\theta)=\int_{\partial \mathcal O} e^{-ik\theta \cdot x} u(x) dS_x.
\end{equation}
Then the following factorization formula (where $\lambda=k^2$) is valid for the far-field operator:
\begin{equation}\label{FFF}
F=\alpha\mathcal L^* (  N^{in}_0-N^{out}) (N^{in}_n - N^{out})^{-1}(N^{in}_0 -N^{in}_n)\mathcal L, \quad \alpha=\frac{1}{4\pi}\left (\frac{k}{2\pi i} \right )^{\frac{d-3}{2}}.
\end{equation}
\end{theorem}
\begin{proof} Let $E=E(x)$ be the solution of the equation $(\Delta+k^2)E=-\delta(x),~x\in \mathbb R^d,$ that satisfies the radiation condition at infinity: $E_r-ikE=o(r^{(1-d)/2}),~r\to\infty$. The scattered wave $\psi_{\text {sc}}$ defined in (\ref{AnoneE}) can be written  as
$$
\psi_{\text {sc}}(x)=\int_{\partial\mathcal O}\left (\frac{\partial E(x-y)}{\partial \nu}\psi_{\text {sc}}-E(x-y)\frac{\partial \psi_{\text {sc}}}{\partial \nu} \right )dS_y, \quad x\in \mathbb R^d\backslash\mathcal O.
$$
The Green formula for functions $E$ and the solution of the Helmholtz equation in $\mathcal O$ with the Dirichlet data $\psi_{\text {sc}}$ at the boundary imply that
$$
\int_{\partial\mathcal O}\frac{\partial E(x-y)}{\partial \nu}\psi_{\text {sc}}dS_y=\int_{\partial\mathcal O}E(x-y) N^{in}_0\psi_{\text {sc}}dS_y, \quad x\in \mathbb R^d\backslash\mathcal O.
$$
Thus
$$
 \psi_{\text {sc}}=\int_{\partial\mathcal O}E(x-y)( N^{in}_0-N^{out})\psi_{\text {sc}}dS_y, \quad x\in \mathbb R^d\backslash\mathcal O,
$$
which leads to the following formula for the scattering amplitude (i.e., for the kernel of operator (\ref{farf})) after passing to the limit $|x| \rightarrow \infty$ :
\begin{equation}\label{scatc2}
f(k,\theta,\omega)  =\alpha\mathcal L^* (  N^{in}_0-N^{out})\psi_{\text {sc}}(x), \quad \alpha=\frac{1}{4\pi}\left (\frac{k}{2\pi i}\right )^{\frac{d-3}{2}}.
\end{equation}

It remains only to find $\psi_{\text {sc}}$ on $\partial\mathcal O$ from (\ref{AnoneE}). Let us denote function $e^{ik\omega \cdot x}$ by $h$. Since $\frac{\partial \psi}{\partial \nu}=\frac{\partial \psi_{\text {sc}}}{\partial \nu}+\frac{\partial h}{\partial \nu}=N^{out}\psi_{\text {sc}}+N^{in}_0h$ and $\frac{\partial \psi}{\partial \nu}=N^{in}_{n}\psi$, the continuity of $\psi$ on $\partial\mathcal O$ implies that
\[
\psi_{\text {sc}}+h=\psi, ~~N^{out}\psi_{\text {sc}}+N^{in}_0h=N^{in}_{n}\psi, \quad x\in \partial\mathcal O.
\]
We apply operator $N^{in}_{n}$ to the first equality and then subtract the second one. This leads to
\[
\psi_{\text {sc}}=(N^{in}_{n}-N^{out})^{-1}(N^{in}_0 -N^{in}_n)h, \quad x\in \partial\mathcal O,
\]
which, together with (\ref{scatc2}), completes the proof.
\end{proof}

The following properties of operators (\ref{lstar}) (compare to \cite[Lemma 5.2]{EP}) will be very essential for us.
\begin{lemma}\label{llstar}
If $-k^2$ is not an eigenvalue of the Dirichlet Laplacian  in $\mathcal O$, then the kernels of operators $\mathcal L^*$ and $\mathcal L$ are trivial, and therefore their ranges are dense.

If $-k^2$ is an eigenvalue of the Dirichlet Laplacian  in $\mathcal O$, then the dimensions of the kernels of operators $\mathcal L^*$ and $\mathcal L$ and their co-kernels (orthogonal complements to the closures of ranges) does not exceed the multiplicity $\kappa$ of the eigenvalue $-k^2$, and $\dim\ker \mathcal L^*=\dim\rm{coker} \mathcal L=\kappa$.
\end{lemma}
\begin{proof}
Let us prove that $\dim\ker \mathcal L^*=\kappa$. It will be done by establishing a one-to-one correspondence between the elements $u$ of the kernel and the eigenfunctions $v$ of the Dirichlet problem. In fact, we will show that the set $\{u\}$ coincides with the set of normal derivatives $-\frac{\partial v}{\partial\nu}$ on $\partial\mathcal O$. The normal derivatives define $v$ uniquely due to the uniqueness of the solution of the Cauchy problem for the Helmholtz equation.

Let $u$ belong to the kernel of $\mathcal L^*$, i.e., $\mathcal L^*u=0$ on $S^{d-1}$. Consider the function
\[
Tu=\int_{\partial\mathcal O}E(x-y)u(y)dS_y, \quad x\in \mathbb R^d,
 \]
 where $E$ is the fundamental solution of the Helmholtz equation that satisfies the radiation conditions. Since $u\in \ker \mathcal L^*$, we have $Tu=o(r^{(1-d)/2}),~r=|x|\to\infty,$ and
 from Rellich's lemma it follows that $ Tu=0$ on  $\mathbb R^d\backslash\mathcal O$. Hence $v:=Tu,~x\in\mathcal O,$ is an eigenfunction of the Dirichlet Laplacian  in $\mathcal O$ with the eigenvalue $-k^2$. Since the density $u$ is equal to the jump of the normal derivative of $Tu$ on $\partial\mathcal O$, we obtain that $u=\frac{\partial v}{\partial\nu}, x\in \partial\mathcal O$.

 Conversely, let $v$ be an eigenfunction of the Dirichlet Laplacian  in $\mathcal O$ with the eigenvalue $-k^2$ and let $u=\frac{\partial v}{\partial\nu}, x\in \partial\mathcal O$. Let $\widehat{v},~ x\in \mathbb R^d,$ be the extension of $v$ to $\mathbb R^d\backslash\mathcal O$ by zero. Then $(\Delta+k^2)\widehat{v}=-\delta (\partial\mathcal O)u,~x\in \mathbb R^d,$ where $\delta (\partial\mathcal O)$ is the delta function on $\partial\mathcal O$. Thus
 \[
 \widehat{v}=\int_{\partial\mathcal O}E(x-y)u(y)dS_y, \quad x\in \mathbb R^d.
 \]
 The right-hand side above has the following asymptotic behavior at infinity:
 \[
 c_d(\mathcal L^*u)(\theta)\frac{e^{ikr}}{r^{(d-1)/2}}(1+o(1)), \quad \theta=x/r, \quad r=|x|\to\infty.
 \]
On the other hand, $\widehat{v}=0$ on $\mathbb R^d\backslash\mathcal O$. Thus, $L^*u=0,$ i.e., $u=\frac{\partial v}{\partial\nu}, x\in \partial\mathcal O,$ belongs to the kernel of operator $L^*.$ Hence $\dim\ker \mathcal L^*=\kappa$. Obviously, $\dim\rm{coker} \mathcal L=\dim\ker \mathcal L^*.$

Since $\dim\ker \mathcal L=\dim\rm{coker} \mathcal L^*$, it remains to show that $\dim\ker \mathcal L\leq\kappa$. Let $\varphi\in \ker \mathcal L$,
\[
\widehat{v}=\int_{S^{d-1}} e^{ikx\cdot \omega} \varphi(\omega) dS_\omega, \quad x\in \mathbb R^d,
\]
and let $v=P\widehat{v},~x\in \mathcal O,$ where $P$ is the restriction on the domain $\mathcal O.$ 
Since $\widehat{v}$ satisfies the equation $\Delta \widehat{v}+k^2\widehat{v}=0$ and $\mathcal L\varphi=0$, function $v$ is an eigenfunction of the Dirichlet Laplacian with the eigenvalue $-k^2$. Since $\widehat{v}$ is the Fourier transform of $\delta(S^{d-1})\varphi$, function $\widehat{v}$ can not vanish identically in $\mathbb R^d$ for non zero $\varphi$. Therefore it can not be equal to zero identically in $\mathcal O$ due to the analyticity in $x$. Hence, the map $\varphi\to v=P\widehat{v}$ is injective, and therefore $\dim\ker \mathcal L\leq\kappa$.
\end{proof}

{\it (C). Additional properties of the D-to-N operators (\ref{NN}).}
Recall that an operator function $A(\lambda):H_1\to H_2,~\lambda \in D,$ in Hilbert spaces $H_1,H_2$ is called Fredholm finitely meromorphic if 1) it is meromorphic in $\lambda\in D,$ 2) it is Fredholm for each $\lambda$ that is not a pole of $A(\lambda)$, 3) if $\lambda=\lambda_0$ is a pole, then the principal part at the pole has a finite range and the analytic part is Fredholm at $\lambda=\lambda_0$.
\begin{lemma}\cite{bleher}\label{blex}
Let $A(\lambda)$ be Fredholm finitely meromorphic in a connected set $D$. If there is a point $\lambda=\lambda_0$ where the operator $A(\lambda)$ is one-to-one and onto, then the operator function $A^{-1}(\lambda)$ is also Fredholm finitely meromorphic.
\end{lemma}
\begin{lemma}\label{b}
Operators (\ref{NN}),
their inverses, and $(N^{in}_0-t)^{-1},~(N^{in}_n-t)^{-1},~ (N^{out}-t)^{-1}$, where $t$ is a constant, are Fredholm finitely meromorphic in $\lambda>0$.

Moreover, operators $(N^{in}_n - N^{out})^{-1}$ and $( N^{out}-t)^{-1}$ are analytic in $\lambda,~ \lambda>0$.
\end{lemma}
To be rigorous, one needs to write $tI$ here and in other similar formulas, but we will omit the identity operator $I$.

\begin{proof} Let us call operators (\ref{NN}), their shifts by $t$, and operator $N^{in}_n - N^{out}$ the direct operators and their inverses the inverse operators. The direct operators are meromorphic in $\lambda$ with finite ranges of residues due to Lemma \ref{mprl}. They are Fredholm since they are elliptic. For each direct operator, one can easily find a point $\lambda=\lambda_0$ where the kernel of the operator is trivial.
Indeed, for the last operator in (\ref{NN}), one can take any $\lambda_0$ that is not an ITE (see Lemma \ref{b}). Such a $\lambda_0$ exists since the set of ITEs is discrete (see \cite{lakvain7}). For other direct operators, except $N^{in}_n - N^{out}$, one can choose any $\lambda_0>0$ that is not an eigenvalue of the corresponding Dirichlet or impedance (with the boundary condition $u_\nu-tu=0$ on $\partial\mathcal O$) problem. Below we show that any $\lambda_0>0$ can be chosen for $N^{in}_n - N^{out}$. Since every elliptic operator on a compact manifold has index zero, all the direct operators are also onto when $\lambda=\lambda_0$, and Lemma \ref{blex} is applicable to these operators. Hence the inverse operators are also Fredholm and finitely meromorphic in $\lambda>0$.

It remains to show that operators $N^{in}_n - N^{out}$ and $N^{out}-t$ do not have kernels when $\lambda>0$ (which implies that their inverse operators are analytic).  Since operator $N^{in}_n$ is symmetric, (\ref{kvf}) implies that
\[
\Im((N^{in}_n - N^{out})f,~f)=-\Im( N^{out}f,~f)<0 \quad {\rm if}~~f\neq 0,
\]
 i.e., operator $N^{in}_n - N^{out},~\lambda>0,$ does not have a kernel. These arguments remain valid if $N^{in}_n$ is replaced by $t$.
\end{proof}

{\it (D). An alternative way to define ITEs}. Due to Lemma \ref{ites}, the factor before $\mathcal L$ in formula (\ref{FFF}) for $F$ establishes a connection between the scattering matrix (related to $F$ by (\ref{sk})) and ITEs. However, the existence of singular ITEs makes it difficult to work with this factorization of $F$ and to use the factor $(N^{in}_0 -N^{in}_n)$, which may have a pole and a kernel at the same ITE. To avoid this difficulty, $F$ will be represented in a different form, where the following operator will be used to relate $F$ and ITEs:
\begin{equation}\label{r2r}
R(\lambda)=(N^{in}_n-t)^{-1}-(N^{in}_0-t)^{-1}:H^s(\partial \mathcal O)\to H^{s+3}(\partial \mathcal O),
\end{equation}
 where $t$ is a constant.  We are going to study operator (\ref{r2r}) now and will discuss the relation between $F$ and $R$ later.

Obviously, boundary conditions (\ref{Antwo}) are equivalent to the following ones:
\begin{equation}\label{AntwoEBB}
\begin{array}{l}
u=v, \quad x \in \partial \mathcal O, \\
\frac{\partial u}{\partial \nu} - tu=\frac{\partial v}{\partial \nu}-tv, \quad x \in \partial \mathcal O,
\end{array}
\end{equation}
where $t>0$ is arbitrary. Consider the set $\{t_s(\lambda)\}$, $\lambda>0$, of values of $t$ for which the impedance problem
\begin{equation}\label{imppr}
-\Delta v - \lambda   n(x)v =0, \quad x \in \mathcal O, \quad v\in H^2(\mathcal O); \quad \frac{\partial v}{\partial \nu}-tv=o, \quad x \in \partial \mathcal O,
\end{equation}
has a non-trivial solution. This set is discrete and countable. Equivalently, one can define this set as the set of values of $t$ for which operator $N_n^{in}(\lambda)-t$ has a non-trivial kernel (see the definition of the kernel in Section 2 (B) if $N_n^{in}(\lambda)-t$ has a pole). One could also describe it as the set of eigenvalues of operator $N_n^{in}(\lambda)$, but it would be a little vague at points $\lambda$ where $N_n^{in}(\lambda)$ has a pole.

The equivalence of two definitions of the set $\{t_s\}$ (via the impedance problem and via the kernel of $N_n^{in}(\lambda)-t$) is obvious when $\lambda$ is not a pole of $N_n^{in}$ (i.e., $\lambda$ is not an eigenvalue of the Dirichlet problem for equation (\ref{Anone})). If $\lambda=\lambda_0$ is a pole of $N_n^{in}(\lambda)$, then it is obvious that the existence of the kernel of  $N_n^{in}(\lambda)-t$ implies the existence of the solution of the impedance problem. The converse statement needs a short justification. We will not provide it since we will not need this converse statement.

We will choose and fix an arbitrary value of $t$ in (\ref{AntwoEBB}) such that for all $s,i$,
\begin{equation}\label{ties}
t\notin \{t_s(\lambda_i^T)\}\bigcup\{t_s(\lambda_i^D)\}\bigcup\{t_s(\lambda_i^{n,D})\}.
\end{equation}
Since we are not going to vary $t$ (except in one insignificant place that will not affect any previous arguments), we usually will not mark explicitly the dependence of any operators or functions on $t$. Thus the value of $t$ in (\ref{r2r}) is fixed.
 \begin{lemma}\label{r5}
 Operator (\ref{r2r}) is meromorphic in $\lambda,~\lambda>0$. It is an elliptic pseudo-differential operator of order $-3$ with the principal symbol $\frac{\lambda(n(x)-1)}{2|\xi^*|^{3}}$. It has a non-trivial kernel only if $\lambda$ is an ITE and the dimension of the kernel coincides with the multiplicity of the ITE. It does not have poles at ITEs.
\end{lemma}
\begin{proof} Operator (\ref{r2r}) can be written as
\[
R(\lambda)=(N^{in}_0-t)^{-1}(N^{in}_0-N^{in}_n)(N^{in}_n-t)^{-1}.
\]
This formula together with Lemmas \ref{mprl}, \ref{b} immediately imply the first two statements of Lemma \ref{r5}. Furthermore, $R(\lambda)\varphi=0$ is equivalent to
\[
(N^{in}_n-t)^{-1}\varphi=(N^{in}_0-t)^{-1}\varphi.
\]
Both sides here are zeroes if and only if $\varphi$ is the Dirichlet data of a singular interior transmission eigenfunction. If they are not zeroes, then $\psi=(N^{in}_0-t)^{-1}\varphi$ satisfies $N^{in}_n\psi=N^{in}_0\psi$ and defines non-singular ITEs. Finally, from (\ref{ties}) it follows that operator  (\ref{r2r}) does not have poles at ITEs.
\end{proof}

We will use the approach to counting ITEs that was developed in \cite{lakvain7}, but now we will use operator (\ref{r2r}) instead of $N^{in}_n-N^{in}_0$. We fix an arbitrary invertible symmetric elliptic operator $D$  of second order defined on $\partial\mathcal O$. Let
 \begin{equation}\label{rhat}
 \widehat{R}(\lambda)=\sigma DRD, \quad
\sigma={\rm sign}_{x\in \partial \mathcal O}(n(x)-1)
  \end{equation}
($\sigma$ is not to be confused with $\sigma _i$ defined in (\ref{defsi})), and let $\{\mu_j(\lambda)\}$ be the set of real eigenvalues of $\widehat{R}(\lambda)$ where $\lambda$ is not a pole of $\widehat{R}(\lambda)$. Let $n^-(\lambda) \geq 0$  be the number of negative eigenvalues $\mu_j(\lambda)$. From Lemma \ref{r5} it follows that $\widehat{R}(\lambda)$ is an elliptic operator of first order with a positive principal symbol. Thus $\mu_j(\lambda)\to\infty$ as $j\to\infty$ and $n^-(\lambda)$ is well defined if $\lambda$ is not a pole of $\widehat{R}(\lambda)$. We will work with operator $\widehat{R}$ instead of $R$ in order to deal with an operator whose eigenvalues converge to infinity, not to zero. On the other hand, operator $D$ establishes a one-to-one correspondence between the kernels of $R$ and the kernels of $\widehat{R}$. Thus ITEs can be defined as values of $\lambda=\lambda^T_i$ where $\widehat{R}(\lambda)$ has a kernel, and $m_i$ eigenvalues of  $\widehat{R}(\lambda)$ vanish at each ITE $\lambda=\lambda^T_i$ of multiplicity $m_i>0$.

Since operator (\ref{r2r}) is self-adjoint and analytic in $\lambda$ in a neighborhood of each ITE, its eigenvalues and eigenfunctions can be chosen to be analytic in these neighborhoods (see \cite[Example 3, XIII.12]{reed}). The following lemma follows from there,  Lemma \ref{r5}, and the theorem on the spectral decomposition of self-adjoint operators (after an appropriate enumeration of the eigenvalues $\mu_j$):
 \begin{lemma}\label{26}
 Let $\lambda=\lambda_0=\lambda_i^T$ be an ITE of order $m_i$. Then there exists $\delta>0$ such that
 \begin{equation}\label{mun1}
\widehat{R}(\lambda)=\sum_{j=1}^{m_i} \mu_j(\lambda)P_{\varphi_j(\lambda)}+K(\lambda), \quad |\lambda-\lambda_0|<\delta,
 \end{equation}
 where $\mu_j(\lambda)$ are analytic (when $|\lambda-\lambda_0|<\delta$) eigenvalues of $\widehat{R}(\lambda)$ such that $\mu(\lambda_0)=0$ and the corresponding eigenfunctions $\varphi_j(\lambda)$ are analytic and orthogonal, $P_{\varphi_j(\lambda) }$ is the projection on $\varphi_j$, and the kernel of $K(\lambda)$ coincides with ${\rm span}\{\varphi_j(\lambda)  \}$.

 The inverse operator has the form:
  \begin{equation}\label{mun2}
\widehat{R}^{-1}(\lambda)=\sum_{j=1}^{m_i} \mu_j^{-1}(\lambda)P_{\varphi_j(\lambda)}+K_1(\lambda), \quad |\lambda-\lambda_0|<\delta,
 \end{equation}
 where $K_1$ is analytic in $\lambda,~ |\lambda-\lambda_0|<\delta$ (it is inverse to $K$ on the subspace orthogonal to ${\rm span}\{\varphi_j(\lambda)  \}$).
 \end{lemma}

Let us denote by $\alpha_i^+,~(\alpha_i^-)$ the number of eigenvalues $\mu_j(\lambda)$ whose Taylor expansion at $\lambda=\lambda^T_i$ starts with an odd power of $\lambda-\lambda^T_i$ and the coefficient for this power has the same sign as $-\sigma$ (respectively, $\sigma$), where $\sigma $ is defined by (\ref{rhat}).
\begin{theorem}\label{th2a}
Let $n(x) \neq 1, x \in \partial \mathcal O$. Then
\begin{equation}\label{theorem2a}
\sum_{i ~ : ~ 0< \lambda^T_i < \lambda}( \alpha^+_i- \alpha^-_i) = \frac{\omega_d}{(2\pi)^d} \gamma   \lambda^{\frac{d}{2}} + O(\lambda^{\frac{d}{2}-\delta}), \quad \lambda \rightarrow \infty,
\end{equation}
where $\delta=\frac{1}{2d}$.
\end{theorem}
\begin{proof} Let us fix an arbitrary point $\alpha>0$ that is not a pole of $\widehat{R}(\lambda)$ and is smaller than the smallest eigenvalues for the Dirichlet problems for equations (\ref{Anone0}) and (\ref{Anone}). Let us evaluate the difference $n^-(\lambda')-n^-(\alpha)$ by moving $\lambda$ from $\lambda=\alpha$ to a value $\lambda=\lambda'>\alpha$. The eigenvalues $\mu_j(\lambda)$ are meromorphic in $\lambda$ and may enter/exit the negative semi-axis $R_\mu^-=\{\mu:\mu<0\}$ only through the end points of the semi-axis. Thus we can split ${n^-}(\lambda')-n^-(\alpha)$ into
\begin{equation}\label{nplus}
{n^-}(\lambda')-n^-(\alpha)=n_1(\lambda')+n_2(\lambda'),
\end{equation}
where $n_1(\lambda')$ is the number of eigenvalues $\mu_j(\lambda)$ that enter/exit the negative semi-axis $R_\mu^-$ through the point $\mu=-\infty$ (when $\lambda$ changes from $\alpha$ to $\lambda'>\alpha$) and $n_2(\lambda')$ is the number of eigenvalues $\mu_j(\lambda)$ that enter/exit the negative semi-axis $R_\mu^-$ through the point $\mu=0$.
Obviously,
\begin{equation}\label{n22}
n_2(\lambda)=\sum_{i ~ : ~ \alpha< \lambda^T_i < \lambda}\sigma( \alpha^+_i- \alpha^-_i).
\end{equation}

Next, one can show that
 \begin{equation}\label{n11}
{n_1}(\lambda')=\sigma(N^t_n(\lambda')-N^t(\lambda')),
\end{equation}
where $N^t_n(\lambda),N^t(\lambda)$ are counting functions for operators $\frac{-1}{n(x)}\Delta$ and $-\Delta$, respectively, with the impedance boundary condition $\frac{\partial u}{\partial \nu}-tu=0$. In order to obtain (\ref{n11}), one needs to note that $\mu_j(\lambda)\to -\infty$ only when $\lambda$ passes through the poles of operator (\ref{r2r}). These poles occur exactly at eigenvalues of the corresponding impedance boundary problem. A rigorous proof of (\ref{n11}) can be obtained exactly as formula (27) in \cite{lakvain7}. The main term in the standard Weyl formula \cite[Th.1.6.1]{safvas} for the counting function of a self-adjoint elliptic problem does not depend on the boundary condition, i.e.,  (\ref{n11}) implies that
\begin{equation}\label{44}
n_1(\lambda) = \frac{\omega_d}{(2\pi)^d}\sigma \gamma   \lambda^{d/2}+O( \lambda^{(d-1)/2}), \quad \lambda\to\infty.
\end{equation}
  An important part of the proof of Theorem \ref{th2a} is the following estimate from above for the number $n^-(\lambda)$ of negative eigenvalues of the operator $R(\lambda)$:
 \begin{equation}\label{n33}
 n^-(\lambda)= O(\lambda^{d/2-\delta}), \quad \lambda\to\infty.
 \end{equation}
The latter estimate can be justified absolutely similarly to an analogous estimate (14) in (\cite{lakvain8}). The statement of the theorem follows immediately from (\ref{n22}), (\ref{44}),  (\ref{n33}).
\end{proof}

We will show in Section 3 that $\alpha^+_i- \alpha^-_i=\sigma_i$. Then Theorem \ref{th} will
follow from Theorem  \ref{th2a}.

{\it (E). An alternative representation of the far-field operator.}
\begin{lemma}\label{lemma1610}
For each positive $\lambda=k^2>0$, operator $\frac{1}{\alpha}F$ can be written as
\begin{equation}\label{lemma0710}
\frac{1}{\alpha}F=Q^*[R_1(\lambda)+R(\lambda)^{-1}+iI(\lambda)]Q, \quad Q=(N^{in}_n-t)^{-1}(N^{in}_n-N^{in}_0)\mathcal L,
\end{equation}
where 1) operators $ R_1, R, I$ are symmetric; 2) operator $R$ is defined in (\ref{r2r});
3) operators $R_1,I$ are analytic in $\lambda$; 4) $R_1$ is an elliptic operator of order one; 5) operator $I$ is infinitely smoothing (has order $-\infty$) and non-negative.  It is strictly positive for all $\lambda>0$, except possibly at most countable set $\{\widehat{\lambda}_s\}$, which does not contain any ITEs $\lambda_i^T$ or eigenvalues of the Dirichlet problem for equation (\ref{Anone}) and does not have any finite limit points.
\end{lemma}
\begin{proof}  We write $F$ (given in Theorem \ref{t21}) in the form
\[
\alpha^{-1}F= \mathcal L^*[(N_0^{in} - N_n^{in})(N^{in}_n - N^{out})^{-1} (N_0^{in} - N_n^{in}) + (N_0^{in} - N_n ^{in})]\mathcal L
\]
\[
=\mathcal L^*(N_0^{in} - N_n^{in})[(N^{in}_n - N^{out})^{-1}+(N_0^{in} - N_n ^{in})^{-1}](N_0^{in} - N_n^{in})\mathcal L=Q^*\widehat{F}Q,
\]
where
\begin{equation}\label{fhfh}
 \widehat{F}=(N^{in}_n-t)[(N^{in}_n - N^{out})^{-1}+(N_0^{in} - N_n ^{in})^{-1}](N^{in}_n-t).
\end{equation}
We need to show that $\widehat{F}=R_1+R^{-1}+ iI$. Let us split the right-hand side in (\ref{fhfh}) into two terms and rearrange the second one. We have
\[
(N^{in}_n-t)(N_0^{in} - N_n ^{in})^{-1}(N^{in}_n-t)=(N^{in}_n-N_0^{in}+N_0^{in}-t)(N_0^{in} - N_n ^{in})^{-1}(N^{in}_n-t)
\]
\[
=-N^{in}_n+t+(N_0^{in}-t)(N_0^{in}-t+t - N_n ^{in})^{-1}(N^{in}_n-t)  =  -N^{in}_n+t+( R)^{-1}.
\]
Hence it remains to show that
\begin{equation}\label{kF1}
F_1:=(N^{in}_n-t)(N^{in}_n - N^{out})^{-1}(N^{in}_n-t)- N^{in}_n+t
\end{equation}
has the form $R_1+iI$, where $R_1, I$ have the properties listed in Lemma \ref{lemma1610}.

One can easily single out the imaginary part of the operator $F_1$:
$$
I(\lambda)=\Im F_1= (N^{in}_n-t)(N^{in}_n - N^{out})^{-1} \Im N^{out} (N^{in}_n - N^{out})^{-1^*}(N^{in}_n-t).
$$
In order to obtain this formula, one can add the factor $(N^{in}_n - N^{out})^*(N^{in}_n - N^{out})^{-1^*}$ after the negative power in expression (\ref{kF1}) for $F_1$ and then use the symmetry of $N^{in}_n$ and the relation $\Im (N^{out})^*=-\Im N^{out}$. Let us justify all the properties of $I(\lambda)$. Operator $(N^{in}_n - N^{out})^{-1} $ is analytic in $\lambda$ due to Lemma \ref{b}. Operator $N^{in}_n$ has poles, but the product $P:=(N^{in}_n-t)(N^{in}_n - N^{out})^{-1}$ is analytic. The easiest way to see the latter property is to replace the first factor in the product $P$ by $(N^{in}_n-N^{out})+(N^{out}-t)$. Thus $I(\lambda)$ is analytic in $\lambda$. The product $P$ is an operator of order zero, and $\Im N^{out}$ has order $-\infty$. The latter follows from (\ref{kvf}),(\ref{FFF}) since the kernel of operator (\ref{FFF}) is infinitely smooth. Thus the order of $I(\lambda)$ is  $-\infty$. Finally,
\[
(I(\lambda)\varphi,\varphi)=\Im(N^{out}\psi,\psi),\quad \psi=(N^{in}_n - N^{out})^{-1*}(N^{in}_n-t)\varphi.
\]
The  latter expression is positive if $\psi\neq 0$ due to (\ref{kvf}).  Thus, in order to obtain the last property of $I(\lambda)$, it remains to find points $\lambda$ where operator $(N^{in}_n - N^{out})^{-1*}(N^{in}_n-t)$ has a non-trivial kernel, i.e., the inverse operator
\[
(N^{in}_n-t)^{-1}(N^{in}_n - N^{out})^*=(N^{in}_n-t)^{-1}(N^{in}_n-t+t - (N^{out})^*)
\]
\[
=I-(N^{in}_n-t)^{-1}( (N^{out})^*-t )
\]
(where $I$ is the identity operator) has a pole. The latter may occur only when $N^{in}_n-t$ has a non-trivial kernel. The corresponding set  $\{\widehat{\lambda}_s\}$ is the set of eigenvalues of the impedance problem (\ref{imppr}) (where $t$ is fixed), and it does not include the ITEs and eigenvalues of the Dirichlet problem due to (\ref{ties}). Hence all the properties of operator $I(\lambda)$ are justified.

To obtain the properties of operator $R_1$, we rewrite (\ref{kF1}) in the form
\[
F_1=(N^{in}_n- N^{out}+ N^{out}-t)(N^{in}_n - N^{out})^{-1}(N^{in}_n-t)- N^{in}_n+t
\]
\[
=(N^{out}-t)(N^{in}_n - N^{out})^{-1}(N^{in}_n-t)=(N^{out}-t)(N^{in}_n - N^{out})^{-1}(N^{in}_n- N^{out}+ N^{out}-t)
\]
\[
=(N^{out}-t)+(N^{out}-t)(N^{in}_n - N^{out})^{-1}( N^{out}-t).
\]
Then Lemmas \ref{mprl} and \ref{b} imply the analyticity of $F_1$. Since $F_1(\lambda)=R_1(\lambda)+iI(\lambda)$ and $I$ is analytic, operator $R_1$ is analytic. From  Lemma \ref{mprl} and formula (\ref{kF1}),
it follows that the principal symbol of $F_1$ is equal to $-|\xi^*|/2$. Thus operator $R_1$ has order one since $I(\lambda)$ is an infinitely smoothing operator.
The proof of Lemma \ref{lemma1610} is complete.
\end{proof}

As we mentioned earlier, it is more convenient for us to work with operator $\widehat{R}$ instead of $R$, and therefore we will use the following version of (\ref{lemma0710}):
\begin{equation}\label{verhat}
\frac{1}{\alpha}F=Q^*D[\widehat{R}_1(\lambda)+\sigma\widehat{R}(\lambda)^{-1}+i\widehat{I}(\lambda)]DQ, \quad Q=(N^{in}_n-t)^{-1}(N^{in}_n-N^{in}_0)\mathcal L,
\end{equation}
where operators $\widehat{R}_1=D^{-1}R_1D^{-1}, ~\widehat{I}=D^{-1}ID^{-1}$ have the same properties as operators $R_1,~I,$ respectively, with the only difference that $\widehat{R}_1$ has order $-3$.

\section{Proof of the main theorems}
{\bf Proof of Theorem \ref{lemmaup}.}
Let $n(x)<1$ on $\partial\mathcal O$ (i.e., $\sigma<0$).  Let $T^+=\overline{{\rm Span}\{\varphi_i^+\}}$, where $\varphi_i^+$ are the eigenfunctions of the scattering matrix $S(k)$ with the eigenvalues $z_i$ in the upper half complex plane $\Im z\geq 0$. In order to prove the first statement of the theorem, we need to show that the space $T^+$ is finite-dimensional.

From (\ref{sk}) it follows that $\Re(\alpha^{-1}F\varphi_i^+,\varphi_i^+)\geq0$. This and the orthogonality of functions $\varphi_i^+$ imply that
\begin{equation}\label{kvf1}
 \Re(\alpha^{-1}F\varphi,\varphi)\geq0, \quad \varphi\in T^+.
\end{equation}
On the other hand, from (\ref{FFF}) and Lemma \ref{mprl} it follows that $\alpha^{-1}F=\mathcal L^*\widehat{F}\mathcal L$, where $\widehat{F}$ is a pseudo-differential operator with the principal symbol $\lambda(n(x)-1)|\xi^*|/2$. For every $\varphi\in H^0(S^{d-1})$, we have
\begin{equation*}
(\alpha^{-1}F\varphi,\varphi)=(\widehat{F}\psi,\psi), \quad \psi=\mathcal L\varphi.
\end{equation*}
 Since $\widehat{F}$ is an elliptic operator of order one with a negative principal symbol, there exists $a>0$ such that
\begin{equation}\label{kvf2}
\Re(\widehat{F}\psi,\psi)\leq -a\|\psi\|^2_{H^{1/2}} +C\|\psi\|^2_{H^0(\partial\mathcal O)},
\end{equation}
and therefore
\begin{equation}\label{kvfMMM}
\Re(\alpha^{-1}F\varphi,\varphi)\leq -a\|\psi\|^2_{H^{1/2}} +C\|\psi\|^2_{L_2(\partial\mathcal O)}, \quad \psi\in\mathcal L T^+.
\end{equation}
From here, (\ref{kvf1}), and the Sobolev imbedding theorem it follows that the set
\[
\overline{\mathcal L T^+}\bigcap\{\|\psi\|_{L_2(\partial\mathcal O)}=1\}
 \]
 is compact in $L_2(\partial\mathcal O)$. Thus, the linear space $\mathcal L T^+$ is finite-dimensional. Now Lemma \ref{llstar} implies that the space $T^+$ is finite-dimensional.

 The first statement of the theorem is proved. To prove the second statement, one needs only to replace $T^+$ by $T^-=\overline{{\rm Span}\{\varphi_i^-\}}$, where $\varphi_i^-$ are the eigenfunctions with the eigenvalues $z_i$ in the lower half complex plane, and use the positivity of the principal symbol of $\widehat{F}$. Let us prove the last statement.

 Assume that the space $T^-={\rm Span}\{\varphi_i^-\}$ is finite-dimensional,
where $\varphi_i^-$ are the eigenfunctions of the scattering matrix $S(k)$ with the eigenvalues $z_i$ in the lower half complex plane $\Im z< 0$.
 Then, similarly to (\ref{kvf1}), we have
 \begin{eqnarray}\nonumber
  \Re(\alpha^{-1}F\varphi,\varphi)\geq0, \quad \varphi\in (T^-)^\bot,  \mbox{ and therefore,}
\\ \label{kvf5}
 \Re(\widehat{F}\psi,\psi)\geq0, \quad \psi\in \mathcal L((T^-)^\bot).
\end{eqnarray}

 We fix an $\varepsilon>0$ so small that the set $\Gamma^-=\partial\mathcal O\bigcap \{x:n(x)<1-\varepsilon\}$ is not empty. Let $n'$ be an infinitely smooth function in $\mathcal O$ such that $0<n'<1$ and $n'$ coincides with $n(x)$ in a $d$-dimensional neighborhood of $\Gamma^-$. From standard local a priori estimates for the solutions of elliptic equations it follows that the operator $G=N^{in}_n-N^{in}_{n'}$ is infinitely smoothing on functions $\psi\in L_2^-$. The latter space consists of functions from $L_2(\partial\mathcal O)$ with the support in $\overline{\Gamma^-}$. Denote by $\widehat{F}'$ operator (\ref{FFF}) with $n$ replaced by $n'$. Since (\ref{kvf2}) holds for  $\widehat{F}'$, it is valid for  $\widehat{F}$ when $\psi\in L_2^-$. This and (\ref{kvf5}) imply that
 \[
 0\leq  -a\|\psi\|^2_{H^{1/2}} +C\|\psi\|^2_{L_2(\partial\mathcal O)},  \quad \psi\in \mathcal L((T^-)^\bot) \bigcap L_2^-.
 \]
The inequality above and the Sobolev imbedding theorem lead to the compactness of the set $\overline{\mathcal L((T^-)^\bot) }\bigcap L_2^-\bigcap\{\|\psi\|_{L_2^-}=1\}$. The  compactness is possible only if the linear space $\overline{\mathcal L((T^-)^\bot)} \bigcap L_2^-$ is finite-dimensional. Since we assumed that $T^-$ is finite-dimensional, it follows that $\overline{\mathcal L(L_2(S^{d-1})) }\bigcap L_2^-$ is finite-dimensional. The latter contradicts Lemma \ref{llstar}. Hence our assumption is wrong, i.e., $T^-$ is infinite-dimensional. Similarly, one can prove that $T^+$ can not be finite-dimensional.

\qed

{\bf Proof of Theorems \ref{th} and \ref{t3}.} {\it Step 1. Quadratic forms related to $S(k)$.} The following general statement plays an important role in the proof of the main results.
Let $0<\alpha_1 < \alpha_2<\pi$. Denote by $S_{\alpha_1,\alpha_2}$ the closed domain in the upper half complex plane bounded by the arc and the chord of the unit circle with the end points at $e^{i\alpha_1},e^{i\alpha_2}$.

\begin{lemma}\label{lemmaold1}
Let a unitary operator $U$ in a Hilbert space $\mathcal H$ have a discrete spectrum. Let $\mathcal H_0 \subset \mathcal H$ be an $m$-dimensional subspace, and let $\mathcal H_1 \subset \mathcal H$ be a subspace of co-dimension $m$. Then the following hold:
\begin{enumerate}
\item The range (the set of values) of the quadratic form $(U\varphi,\varphi), \varphi \in \mathcal H, \|\varphi\|=1,$ coincides with the polygon with the vertices (there may be infinitely many of them) at the eigenvalues of
$U$.
\item If $(U\varphi,\varphi)\in S_{\alpha_1,\alpha_2}$ for each $\varphi \in \mathcal H_0, \|\varphi\|=1,$
then $U$ has at least $m$ eigenvalues $z$ (with the multiplicities taken into account) with ${\rm arg}z\in(\alpha_1,\alpha_2)$.
\item If  $(U\varphi,\varphi), \varphi \in \mathcal H_1, \|\varphi\|=1,$ does not have values in $S_{\alpha_1,\alpha_2}$,
then $U$ has at most $m$ eigenvalues $z$ (with the multiplicities taken into account) with ${\rm arg}z\in(\alpha_1,\alpha_2)$.
\end{enumerate}
\end{lemma}
{\bf Proof.} The form has values
$$
\sum_{i} t^2_i e^{i\gamma_i}, \mbox{ where } \sum t^2_i =1.
$$
Here $e^{i \gamma_i}$ are the eigenvalues of $U$. This implies the first statement. If the second assumption holds, then the existence of at least one eigenvalue follows immediately from the first statement. If there are only $m_1<m$ linearly independent normalized eigenfunctions $\varphi_j, 1\leq j\leq m_1,$ with eigenvalues on the arc that bounds $S_{\alpha_1,\alpha_2}$, then one can apply the first statement of the lemma to the same quadratic form on the space orthogonal to ${\rm span}\varphi_j$ and prove the existence of one more eigenvalue on the same arc.

Let us prove the last statement. Assume that $U$ has more than $m$ eigenvalues with  ${\rm arg}z\in(\alpha_1,\alpha_2)$. Then there exists a linear combination $\varphi$ of the corresponding eigenfunctions that belongs to $\mathcal H_1$. Since $(U\varphi,\varphi) \in S_{\alpha_1,\alpha_2}$, we arrive at a contradiction, which proves the statement.

\qed

The next lemma contains some statements on relations between the far-field operator $F$ (see (\ref{farf})) and the eigenvalues $z_j(k)$ of the scattering matrix $S(k)$.
 {\begin{lemma}\label{lll} (A) Let $n(x)<1$  on $\partial \mathcal O$. Then the following hold:

1) If there exist $m^\pm$-dimensional subspaces $\Phi^\pm=\Phi^\pm(k)$ in $L_2(S^{d-1})$ such that the following relations hold for the far-field operator $F$  when $k\to k_0\mp0$,
\begin{equation}\label{fh2}
0<{\rm arg}(\alpha^{-1}F\varphi,\varphi)<\delta(k),~~0\neq\varphi\in \Phi^\pm,~~ where~~\lim_{k\to k_0\mp 0}\delta(k)=0,
\end{equation}
then the scattering matrix $S(k),~\pm(k_0-k)>0,$ has at least $m^\pm$ eigenvalues $z_j(k)=z_j^\pm(k), 1\leq j\leq m^\pm,$ on $C_+$ that approach $z=1$ moving  clockwise (counterclockwise, respectively), i.e.,
\[
\lim_{k\to k_0\mp 0}z_j(k)=1\! + \!i0.
\]

2) If there exist subspaces $\Phi'_\pm=\Phi_\pm'(k)$ in $L_2(S^{d-1})$ of co-dimensions $m^\pm$ such that
\begin{equation}\label{abs}
{\rm arg}(\alpha^{-1}F\psi,\psi)\notin (0,\delta),\quad 0\neq \psi \in\Phi'_\pm,\quad \varepsilon>\pm(k_0-k)>0,
\end{equation}
then the scattering matrix $S(k)$ has at most $m^\pm$ eigenvalues on the arc $0<{\rm arg}z<\delta$ of the unit circle when $\varepsilon>\pm(k_0-k)>0$.

3) Thus if both assumptions 1) and 2) hold, then $S(k)$ has $m^+~(m^-)$ eigenvalues on $C_+$ that approach $z=1$ moving clockwise (counterclockwise, respectively) when $k\to k_0\mp 0$, and all other eigenvalues on $C_+$ are separated from $z=1$ when $k$ is close enough to $k_0$.

(B) Let $n(x)>1$  on $\partial \mathcal O$. Then the following hold:

1')  If there exist $m^\pm$-dimensional subspaces $\Phi^\pm=\Phi^\pm(k)\subset L_2(S^{d-1})$ such that
\begin{equation}\label{absab}
\pi-\delta(k)<{\rm arg}(\alpha^{-1}F\psi,\psi)<\pi,~~0\neq\varphi\in \Phi^\pm,~~ where~~\lim_{k\to k_0\pm 0}\delta(k)=0,
\end{equation}
then the scattering matrix $S(k)$ has at least $m^\pm$ eigenvalues on $C_-$ that approach $z=1$ moving clockwise (counterclockwise, respectively).

2') If there exist subspaces $\Phi'_\pm=\Phi_\pm'(k)\subset L_2(S^{d-1})$ of co-dimension $m^\pm$ such that
\begin{equation*}
{\rm arg}(\alpha^{-1}F\psi,\psi)\notin(\pi-\delta,\pi),\quad 0\neq \psi \in\Phi'_\pm,\quad \varepsilon>\pm(k-k_0)>0,
\end{equation*}
then the scattering matrix $S(k)$ has at most $m^\pm$ eigenvalues on the arc $\pi-\delta<{\rm arg}z<\pi$ of the unit circle when $\varepsilon>\pm(k-k_0)>0$.
\end{lemma}

{\bf Proof.} The eigenvalues of the unitary operator $S(k)$ belong to the unit circle. Therefore, from (\ref{sk}) it follows that the eigenvalues of the operator $ \alpha^{-1}F(k)$ belong to the circle of radius $1/(2k|\alpha|^2)$  centered at $1/(2ik|\alpha|^2)$ and, moreover, if $\sigma=-1$, then the values of the quadratic form $(S(k)\varphi,\varphi)$ with $\|\varphi\|=1$ belong to the set $S_{0,\gamma}$ if and only if (\ref{fh2}) holds with $\delta(k)=\gamma.$ Similarly, if $\sigma=1$, then the values of the quadratic form $(S(k)\varphi,\varphi),~\|\varphi\|=1,$ belong to the set $S_{-\gamma,0}$ if and only if (\ref{absab}) holds with $\delta(k)=\gamma.$ Thus Lemma \ref{lll} is the direct consequence of Lemma \ref{lemmaold1}.

\qed

{\it Step 2. Plan to complete the proofs of Theorems \ref{th} and \ref{t3}.} Let $\lambda_0=\lambda_i^T$ be an ITE of multiplicity $m_i$, and let $\beta_i^\pm$ be the number of eigenvalues $\mu_j,~j\leq m_i,$ in formulas (\ref{mun1}), (\ref{mun2})  that are negative when $1 \gg \varepsilon>\pm\sigma(\lambda-\lambda_0)>0$. Let us stress that we consider only those $\mu_j(\lambda)$ that vanish at $\lambda=\lambda_0$. Obviously, $\beta_i^\pm=\alpha_i^\pm+r$, where $\alpha_i^\pm$ are defined in (\ref{theorem2a}) and $r$ is the number of eigenvalues whose Taylor expansion starts with an even positive power of $\lambda-\lambda_0$ and has a negative coefficient for this power. In particular,
\begin{equation}\label{alphab}
\beta_i^+-\beta_i^-=\alpha_i^+-\alpha_i^-.
\end{equation}

 In Step 3, we are going to show that there exists a $\beta_i^+$-dimensional subspace $\Phi=\Phi^+(\lambda)$ in $L_2(S^{d-1})$ on which (\ref{fh2}) holds when $k\to k_0-0$ and there is a $\beta_i^-$-dimensional subspace $\Phi=\Phi^-(\lambda)$ in $L_2(S^{d-1})$ on which (\ref{fh2}) holds when $k\to k_0+0$. We refer to relation (\ref{fh2}) below, but in fact we are going to justify simultaneously (\ref{fh2}) when $\sigma<0$ and its analogue for $\sigma>0$ stated in part (B) of Lemma \ref{lll}. In Step 4, we will prove (\ref{abs}) with $\Phi'=\Phi'_\pm$ of co-dimension $\beta^\pm_i$ when $k\to k_0\mp 0$ (and its analogue from part (B) of the same lemma). Then Lemma  \ref{lll} will justify all the statements of Theorem \ref{t3} with $m^\pm_i=\beta^\pm_i$. In particular, the last statement of the theorem will be justified because the arguments below are valid when $m_i=0$, i.e., $\lambda=\lambda_0$ is not an ITE (more details will be given in Step 4). Since the relation $m^\pm_i=\beta^\pm_i$ will be established, (\ref{alphab}) will imply that $\alpha_i^+-\alpha_i^-=m_i^+-m_i^-$. Thus Theorem \ref{th} will be a consequence of Theorem \ref{th2a}. Hence the proofs of Theorems \ref{th} and \ref{t3} will be completed as soon as (\ref{fh2}) and (\ref{abs}) are established.

 Let us make one more remark concerning the next steps. The ITEs are values of $\lambda$, and the Dirichlet-to-Neumann maps are functions of $\lambda$, while it is customary to consider the far-field operator and the scattering matrix as functions of $k$.  Many formulas below will contain simultaneously $k$ and $\lambda$. It will always be assumed (without reminders) that $\lambda=k^2$.

{\it Step 3. Establishing (\ref{fh2}).}
We will consider only the case of $\sigma(\lambda-\lambda_0)\to + 0$ since the arguments in the case of $\sigma(\lambda-\lambda_0)\to - 0$ are no different (only the pluses in the indices must be replaced by minuses in the latter case).

Denote by $\widehat{F}$ the operator in the square brackets in the right-hand side of formula (\ref{verhat}). Let
\begin{equation}\label{fh5}
 \widehat{\Phi}^+=\widehat{\Phi}^+(\lambda):={\rm span}\{\varphi_j,~1\leq j\leq \beta_i^+\},
\end{equation}
where $\varphi_j=\varphi_j(\lambda)$ are functions defined in (\ref{mun1}). The enumeration is such that the functions in  (\ref{mun1}) with $\mu_j(\lambda)<0$  when $\sigma(\lambda-\lambda_0)\to + 0$ are listed first. Then from (\ref{verhat}) it follows that
\begin{equation}\label{fh}
(\widehat{F}\varphi,\varphi)=\sigma\sum_{j=1}^{\beta_i^+}c_j^2\mu_j^{-1}(\lambda)+O(1), \quad \varphi=\sum_{j=1}^{\beta_i^+}c_j\varphi_j\in \widehat{\Phi}^+, \quad \sigma(\lambda-\lambda_0)\to + 0.
\end{equation}
This implies that $|\Im (\widehat{F}\varphi,\varphi)|=O(1)$ and
\begin{equation}\label{sss}
\frac{|\Im (\widehat{F}\varphi,\varphi)|}{\Re(\widehat{F}\varphi,\varphi)}\to 0,~~\sigma(\lambda-\lambda_0)\to + 0, \quad 0\neq \varphi\in \Phi^-.
\end{equation}

The imaginary part of the form  (\ref{fh}) is positive (due to Lemma \ref{lemma1610}), and the real part has the same sign as $-\sigma$. Thus (\ref{sss}) justifies (\ref{fh2}) for $\sigma<0$ and its analogue for $\sigma>0$ from part (B) of Lemma \ref{lll}, but both relations are justified for the operator $\widehat{F}$ (on the space $ \widehat{\Phi}^-$) instead of the operator $\alpha^{-1}F$.
 Then the same relations for
$\widehat{F}$ hold for an arbitrary  $\beta_i^+$-dimensional subspace $\widehat{\Phi}^+_\varepsilon(\lambda)$ in $L_2(S^{d-1})$ if it is close enough to $\widehat{\Phi}^+(\lambda)$, where the
distance between these subspaces may depend on $\lambda$. Since operator $DQ$ for each $\lambda\in(\lambda_0-\varepsilon,\lambda_0)$ and small enough $\varepsilon$ has a dense range (see Lemma \ref{llstar}), one can find functions $\psi_j$ such that
$DQ\psi_j$ are so close to $\varphi_j$ that (\ref{fh2}) holds for operator $\widehat{F}$ on the subspace $\widehat{\Phi}^+_\varepsilon(\lambda)
={\rm span}\{DQ \psi_j\}$. Then (\ref{fh2}) and its analogue for $\sigma>0$ hold for $\alpha^{-1}F$ with $\Phi^+={\rm span}\psi_j$.

{\it Step 4. Establishing (\ref{abs}).} As in the previous step, we could prove simultaneously (\ref{abs}) and its analogue for $\sigma>0$. However, we will assume that $\sigma<0$ to make the text more transparent. Also, we are going to consider only the case of $\lambda<\lambda_0=\lambda_i^T$ since the case of $\lambda>\lambda_0$ is no different.

Let us show that (\ref{abs}) holds with  $\Phi'= (Q^* D\widehat{\Phi}^+)^\bot$, where $\widehat{\Phi}^+$ is defined in (\ref{fh5}).
Due to Lemma \ref{llstar}, it is possible to choose $\varepsilon>0$ small enough so that the kernel of the operator $Q^* D$ is trivial when $\lambda_0-\varepsilon<\lambda<\lambda_0$, and therefore the dimension
of $ Q^* D\widehat{\Phi}^+$ is $\beta^+_i$.

If $\psi\in\Phi'$, then $\varphi:=Q\psi$ is smooth enough since $Q$ contains the factor $\mathcal L$, which is an infinitely smoothing operator. In particular, $\varphi\in H^{1}(\partial \mathcal O)$. Furthermore, $D\varphi\bot\Phi^+$, and (due to (\ref{lemma0710}))
\[
(\alpha^{-1}F\psi,\psi)=([R_1(\lambda)+R(\lambda)^{-1}+iI(\lambda)]\varphi,\varphi).
\]
 Hence it is enough to show that
\begin{equation}\label{ad}
{\rm arg}([R_1+ R^{-1}+iI]\varphi,\varphi)\notin (0,\delta)\quad {\rm when} ~~ \lambda_0-\varepsilon<\lambda<\lambda_0
\end{equation}
for smooth functions $\varphi\neq 0$ such that $D\varphi\bot\widehat{\Phi}^+$. Obviously, it is enough to consider smooth functions $\varphi\neq 0$ from the space
\[
\Phi_1=(D\widehat{\Phi}^+)^\bot\bigcap\{ \|\varphi\|_{H^{1}(\partial \mathcal O)}=1\}\bigcap\{\Re(\widehat{F}\varphi,\varphi)>0\},
\]
and  (\ref{ad}) will be proved if we show the existence of constants $\gamma_1,\gamma_2>0 $ such that the following estimates are valid for the real and imaginary parts of the form (\ref{ad}):
\begin{equation}\label{25}
([R_1+R^{-1}]\varphi,\varphi)<\gamma_1, \quad (I\varphi,\varphi)>\gamma_2>0~~~{\rm for}~~\varphi\in \Phi_1.
\end{equation}

From (\ref{rhat}) and (\ref{mun2}) it follows that
\[
(R^{-1}\varphi,\varphi)=\sigma(\widehat{R}^{-1}D\varphi,D\varphi)=\sum_{j=\beta_i^-+1}^{m_i}\frac{\sigma}{\mu_j(\lambda)}\|P_{\varphi_j}D\varphi\|^2+\sigma(K_1D\varphi,D\varphi)
\]
\[
\leq(\sigma K_1(\lambda)D\varphi,D\varphi).
\]
We omitted the terms under the summation sign here since $\mu_j(\lambda)>0$ (and $\sigma\mu_j(\lambda)<0$) when $j>\beta_i^+,~\lambda_0-\varepsilon<\lambda<\lambda_0$. For each $\lambda\neq \lambda_0$, operator $\sigma\widehat{R}(\lambda)$ is an elliptic pseudo-differential
operator of first order with a negative principal symbol (see Lemma \ref{r5} and formula (\ref{rhat})), and therefore $\widehat{R}^{-1}(\lambda)$ is an elliptic pseudo-differential operator of order $-1$ with a negative
principal symbol. Operator $K_1(\lambda)$ differs from $\widehat{R}^{-1}(\lambda)$  by a projection on a finite-dimensional space spanned by $C^\infty$ functions. Thus it has the same
properties as  $\widehat{R}^{-1}(\lambda)$, but additionally it is analytic in $\lambda$ (see Lemma \ref{26}). Hence there is a constant $a_1>0$ such that
\[
(R^{-1}(\lambda)\varphi,\varphi)\leq (\sigma K_1(\lambda)D\varphi,D\varphi)\leq -a_1\|D\varphi\|^2_{H^{-1/2}}+O(\|D\varphi\|^2_{H^{-1}})
\]
\[
\leq -a_1\|\varphi\|^2_{H^{3/2}(\partial \mathcal O)}+O(\|\varphi\|^2_{H^{1}(\partial \mathcal O)}), \quad \lambda_0-\varepsilon<\lambda<\lambda_0.
\]
Operator $R_1$ is an elliptic operator of first order, and it is analytic in
$\lambda $ in a neighborhood of $\lambda_0 $ (see Lemma \ref{lemma1610}). Thus
\[
([R_1+R^{-1}]\varphi,\varphi)\leq -a_1\|\varphi\|^2_{H^{3/2}}+a_2(\|\varphi\|^2_{H^{1}}), \quad \lambda_0-\varepsilon<\lambda<\lambda_0.
\]
This implies the first estimate in (\ref{25}) and also the compactness of the set $\Phi_1$ in $H^1(\partial \mathcal O)$. Indeed, since $\|\varphi\|_{H^1}=1$ in $\Phi_1$,
from the line above it follows that $\Re(\widehat{F}\varphi,\varphi)>0$ on $\Phi_1$ only if $\|\varphi\|_{H^{3/2}}$ is bounded. Thus the set $\Phi_1$ is compact in $H^1(\partial \mathcal O)$ due to the
Sobolev imbedding theorem.

Further, due to (\ref{kvf}), $\Im(\widehat{F}\varphi,\varphi)>0$ on each element $0\neq \varphi\in H^{1}(\partial \mathcal O)$ for $\lambda_0-\varepsilon/2\leq\lambda\leq\lambda_0$ (the end points are included). Then
the compactness of $\Phi_1$ in $H^{1}(\partial \mathcal O)$ implies that $\Im(\widehat{F}\varphi,\varphi)$ has a positive lower bound on $\Phi_1$, i.e., the second estimate in (\ref{25}) holds. Thus  (\ref{abs}) is justified.

 {\it Step 5.} All the arguments in the previous step, used to prove (\ref{abs}), are valid when $m_i=0$ (i.e., for $\lambda=\lambda_0$, which is not an ITE) if $\lambda_0$ does not belong to the exceptional set $\{\widehat{\lambda}_s\}$ defined in
Lemma \ref{lemma1610}. This set is discrete and consists of eigenvalues of the impedance problem (\ref{imppr}) where $t$ was fixed at an earlier stage (see the proof of Lemma \ref{lemma1610}). After  (\ref{abs}) is proved for all $\lambda$ except a fixed exceptional set $\{\widehat{\lambda}_s\}$, we can change the value of $t$ to another value $t=t_s$ for which $\lambda=\widehat{\lambda}_s$ is not an eigenvalue of the impedance problem (\ref{imppr}) with $t=t_s$. Then  (\ref{abs}) will be justified for $\lambda=\widehat{\lambda}_s$. In fact, we can find a value of $t=\widehat{t}$ that can be used simultaneously for all points  $\widehat{\lambda}_s$, but we do not need to do it.

The proof of Theorems  \ref{th} and \ref{t3} is complete.

\qed

{\bf Proof of Theorem \ref{th14}.} First, let us show that the eigenvalues of the operator $R(\lambda)$, defined in (\ref{rhat}), can only have simple zeroes. Indeed, if $\lambda>0$ is not a pole of the operator $R^{-1}(\lambda)$, then $R^{-1}(\lambda)$ maps an arbitrary function $f \in H^{3/2}(\partial\mathcal O)$ into
\begin{equation}\label{1212}
R^{-1}(\lambda)f=(\frac{\partial u}{\partial \nu} - tu)|_{x\in \partial\mathcal O} =  (\frac{\partial v}{\partial \nu} - tv)|_{x\in \partial\mathcal O},
\end{equation}
where $(u,v)$ is the solution of the problem
\begin{eqnarray}\label{AnoneM}
\Delta u +\lambda u=0,  \quad u\in H^2(\mathcal O), \\ \label{AnoneM1}
\Delta v  + \lambda n(x)v=0, \quad v\in H^2(\mathcal O), \\ \label{AntwoM2}
\begin{array}{l}
u-v=f, \quad x \in \partial \mathcal O, \\
\frac{\partial u}{\partial \nu} - tu =  \frac{\partial v}{\partial \nu} - tv, \quad x \in \partial \mathcal O.
\end{array}
\end{eqnarray}
One can express solution $(u,v)$
of (\ref{AnoneM})-(\ref{AntwoM2}) through the resolvent of the ITE problem by looking for $(u,v)$ as a sum of  two terms, where the first term $(u_1,v_1)$ satisfies only the boundary conditions, and the second term is the solution of problem (\ref{AnoneM})-(\ref{AntwoM2}) with homogeneous boundary conditions and the right-hand side in the equations defined by the first term. Hence the operator $f\to (u,v)$ has a pole of at most first order at $\lambda=\lambda_0$ if the resolvent of the ITE problem has a pole of first order at $\lambda_0$.
Therefore, (\ref{1212}) implies that the eigenvalues of the operator $R(\lambda)$ may have zeroes only of first order at $\lambda=\lambda_0$.

Now let $\lambda=\lambda_0$ be an ITE, and let $\{\varphi_j(\lambda)\}$ be an analytic in $\lambda,~ |\lambda-\lambda_0|\ll 1,$ orthonormal system of eigenfunctions of the operator $\widehat{R}(\lambda)$ with the eigenvalues $\mu_j(\lambda),~\mu_j(\lambda_0)=0,$ defined in Lemma \ref{26}.  Such a system exists \cite[Example 3, XIII.12]{reed} for an arbitrary self-adjoint and analytic family of  operators when $\lambda_0$ is an isolated eigenvalue of finite multiplicity. Formula (\ref{rhat}) implies that functions $\{\psi_j=D\varphi_j(\lambda_0)\}$ form a basis in the kernel of operator $R(\lambda_0)$.  Since $\|\varphi_j\| \equiv 1$, we have that}
 \begin{equation}\label{lastlast}
 \mu_j'(\lambda_0) = (\widehat{R}'(\lambda)\varphi_j,\varphi_j)|_{\lambda=\lambda_0}=\sigma (R'(\lambda)\psi_j,\psi_j)|_{\lambda=\lambda_0}, \quad 1 \leq j \leq m_i.
\end{equation}

Further, due to (\ref{1212})-(\ref{AntwoM2}), the functions $\psi_j=\psi_j(\lambda_0)$ in the kernel of $R(\lambda_0)$ are the impedance  values (\ref{AntwoM2}) of the components of the eigenfunctions $(u_j,v_j)$ of the ITE problem with the eigenvalue $\lambda=\lambda_0$, i.e. (for transparency, we omit index $j$ below),
\begin{equation}\label{AA1}
\psi=(\frac{\partial u}{\partial \nu} - tu)|_{x\in \partial\mathcal O} =  (\frac{\partial v}{\partial \nu} - tv)|_{x\in \partial\mathcal O}, \quad \lambda=\lambda_0.
\end{equation}
Let $u(\lambda)$ and $v(\lambda)$ be the solutions of (\ref{AnoneM}) and (\ref{AnoneM1}), respectively, with the boundary conditions (\ref{AA1}).  Using the Green formula  and the fact that $\frac{\partial v' }{ \partial \nu} -tv'=0$ at the boundary and $\Delta v' + \lambda n v'=-nv$ in the domain, we obtain that
\begin{equation*}
\frac{d}{d\lambda} ((N^{in}_n(\lambda) - t)^{-1}\psi,\psi)= \int_{\partial \mathcal O} v'\overline{\psi}dS=
\int_{\partial \mathcal O} v'\overline{\psi}dS- \int_{\partial \mathcal O} (\frac{\partial v' }{ \partial \nu} - tv')\overline{v}dS
\end{equation*}
\begin{equation}\label{0410B}
=
\int_{\mathcal O}v'(\overline{\Delta v +\lambda n(x) v})-\int_{\mathcal O} (\Delta v'+\lambda n(x) v')\overline{v} = \int_{\mathcal O} n(x)|v|^2dx.
\end{equation}
A similar relation (with $v$ replaced by $u$) is valid when $n(x)\equiv 1$. Thus
$$
\nu'(\lambda_0) = (\frac{d}{d\lambda}((N^{in}_n- t)^{-1}- (N^{in}_0 - t)^{-1})(\lambda)\psi,\psi)|_{\lambda=\lambda_0} =  \int_{\mathcal O} (n|v|^2-|u|^2)dx,
$$
and therefore (see the last sentence of the first paragraph of the proof)
\begin{equation}\label{l4}
\mu_j'(\lambda_0) = \sigma \int_{\mathcal O} (n|v_j|^2-|u_j|^2)dx\neq 0, \quad 1\leq j\leq m_i.
\end{equation}

Let us show that $\int_{\mathcal O} (nv_i \overline{v}_j- u_i \overline{u}_j)dx = 0$ if $i\neq j, ~\lambda=\lambda_0$. Indeed, similarly to (\ref{0410B}) we obtain
$$
\int_{\mathcal O} (nv_i \overline{v}_j- u_i \overline{u}_j)dx = \sigma(R'(\lambda_0)\psi_i(\lambda_0),\psi_j(\lambda_0)).
$$
The right-hand side can be written as
$$
\sigma(R'(\lambda_0)\psi_i(\lambda_0),\psi_j(\lambda_0))= (\widehat{R}'(\lambda_0)\varphi_i(\lambda_0),\varphi_j(\lambda_0))
$$
$$
=(\widehat{R}(\lambda)\varphi_i(\lambda),\varphi_j(\lambda))'|_{\lambda=\lambda_0}- [(\varphi_i'(\lambda_0),\widehat{R}(\lambda_0)\varphi_j(\lambda_0))+(\widehat{R}(\lambda_0)\varphi_i(\lambda_0),\varphi_j'(\lambda_0))].
$$
The first term here is zero due to the orthogonality of $\varphi_i$ and $\varphi_j$, and the second term is zero since the
eigenvalues of functions $\varphi_i,\varphi_j$ vanish at $\lambda=\lambda_0$.

 Since the signature of the form does not depend on the choice of the basis, we have
\[
{\rm sgn}A=-\sigma\sum_{j=1}^{m_i}{\rm sign}\mu_j'(\lambda_0).
\]
Since $\mu_j'(\lambda_0)\neq 0$, from the definition of $\beta^\pm_i$ given in Step 2 of the proof of Theorem \ref{t3} it follows that the right-hand side in the formula above is equal to $\beta^+_i-\beta^-_i$. It was also shown in the proof of Theorem \ref{t3} that $m^\pm_i=\beta^\pm_i$. This completes the proof of Theorem~\ref{th14}.

\qed

Let us provide some arguments supporting the conjecture stated at the end of the introduction. For simplicity let us assume that there is a unique Jordan block of size $m>1$ corresponding to an ITE $\lambda^T_i$. The relation between the operator $R^{-1}(\lambda)$ and the resolvent of the ITE problem that was established at the beginning of the proof of Theorem \ref{th14} implies that there is a single eigenvalue $\mu(\lambda)$ of the operator $\widehat{R}(\lambda)$  that vanishes at $\lambda=\lambda^T_i$, and this eigenvalue has zero of order $m $ at $\lambda=\lambda^T_i$.
Then from the proof of Theorem \ref{t3} it follows that $\sigma_i=\pm 1$ if $m $ is odd and $\sigma_i=0$ if $m$ is even. Since the ITE problem is symmetric with respect to the indefinite metric  $J(u,v)= \int_{\mathcal O} \left (|u|^2- n|v|^2 \right ) dx$, the same relation is valid for the signature of the matrix $A$, see \cite{Gohberg}.

\end{document}